\documentclass[final,5p,times,twocolumn]{elsarticle}
\usepackage{url}
\usepackage{subfig}
\usepackage{graphicx,dblfloatfix}
\usepackage{caption}
\usepackage{color}
\usepackage{hyperref}
\usepackage{amssymb}
\usepackage{amsmath}
\usepackage{amsthm}
\usepackage{algorithm} 
\usepackage{algpseudocode}
\usepackage{enumerate}
\usepackage{tabularx}
\usepackage{multirow}
\usepackage{setspace}

\biboptions{sort&compress}

\newtheorem{thm}{Theorem}

\newproof{pf}{Proof}

\algrenewcomment[1]{/*#1*/}
\algnewcommand{\LineComment}[1]{\State /*#1*/}
  
\journal{Journal of Systems Architecture}

\begin{document}

\begin{frontmatter}

\title{A Federated Deep Learning Framework for \\ Privacy Preservation and Communication Efficiency}   

\author[1]{Tien-Dung Cao}
\ead{dung.cao@ttu.edu.vn}
\author[2]{Tram Truong-Huu\corref{cor1}}
\ead{truonghuu.tram@singaporetech.edu.sg}
\author[1]{Hien Tran}
\ead{hien.tran@ttu.edu.vn}
\author[3]{Khanh Tran\corref{cor2}}
\ead{tvkhanh@hcmiu.edu.vn}

\cortext[cor1]{Corresponding author}
\cortext[cor2]{This work has been done when Khanh Tran was with Tan Tao University, Long An, Vietnam}

\address[1]{School of Engineering, Tan Tao University, Long An Province, Vietnam} 
\address[2]{Singapore Institute of Technology, Singapore}
\address[3]{Department of Mathematics, International University, Vietnam National University-Ho Chi Minh City, Ho Chi Minh City, Vietnam}

\begin{abstract}
Deep learning has achieved great success in many applications. However, its deployment in practice has been hurdled by two issues: the privacy of data that has to be aggregated centrally for model training and high communication overhead due to transmission of large amount of data usually geographically distributed. Addressing both issues is challenging and most existing works could not provide an efficient solution. In this paper, we develop FedPC, a Federated Deep Learning Framework for Privacy Preservation and Communication Efficiency. The framework allows a model to be learned on multiple private datasets while not revealing any information of training data, even with intermediate data. The framework also minimizes the amount of data exchanged to update the model. We formally prove the convergence of the learning model when training with FedPC and its privacy-preserving property. We perform extensive experiments to evaluate the performance of FedPC in terms of the approximation to the upper-bound performance (when training centrally) and communication overhead. The results show that FedPC maintains the performance approximation of the models within $8.5\%$ of the centrally-trained models when data is distributed to 10 computing nodes. FedPC also reduces the communication overhead by up to $42.20\%$ compared to existing works.
\end{abstract}

\end{frontmatter}

\section{Introduction}
\label{sec:intro}
The development of information technology makes the data collection and storage process much easier than ever. This leads to the explosion of the amount of data collected every day in any daily life aspects such as healthcare, business, manufacturing, security, to name a few. Analyzing such data provides a lot of useful insights that could help one anticipate business or reaction plans to improve their ultimate objectives or prevent damages due to negative events. Obviously, analyzing such a huge amount of data requires tools and software that automate the analytic process and achieve high throughput at least equal to the velocity of data collection so as to provide the analytic results as soon as possible.

Generally, in most of real-world applications, data, especially personal information, is generated and stored in data silos, either end-users' devices or service providers' data center. Most conventional machine learning and deep learning are trained in a centralized manner, requiring training data to be fused in a data server. In the era of Internet of Things (IoT), collecting and aggregating data from multiple sources incur significant communication overhead and network resource cost. It is even more difficult with resource-constrained devices such as IoT devices or in the area where network resources are limited such as developing countries.

The challenges are not only in the problem of transporting high-volume, high-velocity, high-veracity, and heterogeneous data across organisations but also in the data privacy issue as the data owners do not want their data to leave their premises, especially the data contains sensitive information such as medical records, bank transactions, security logs, etc. This privacy concern prevents the data owners from contributing their data to the training process even though they might know that their data could improve the model performance. This motivates us to develop a novel training framework that allows the model to be trained on different private datasets without relocating/gathering them to the same location. Without trusting any third parties including training coordinator and data owners, such a training approach needs to ensure that there is not any sensitive data leaked, thus preventing a data owner or training coordinator from inferring the data of other owners.  

Several works have addressed the above challenges~\cite{fedavg2017,phong2019}. However, these works either chosen to address only one of the challenges or compromised one for the other. For instance, Phong and Phuong~\cite{phong2019} developed a privacy training algorithm that preserves the privacy of each private dataset in the ensemble training dataset but requiring a sequential training process through all the private datasets and model is repeatedly copied from one computing node to another computing node. Combining the two aforementioned challenges in an integrated framework is a hard problem that cannot be achieved just by adopting or combining the existing solutions altogether. Protecting the data privacy by keeping the data in the respective premise of the data owner requires a training framework that allows the data owner to exchange the model parameters learned with the data with the training coordinator that maintains a global model instance. Such a framework has to minimize the communication overhead, i.e., the amount of data exchanged, idle time of the training processes on different private datasets, and avoid revealing additional information about the data. 

In this paper, we design and develop \textbf{FedPC}, a \textbf{Fed}erated Deep Learning Framework for \textbf{P}rivacy Preservation and \textbf{C}ommunication Efficiency. The framework addresses the above challenges in an efficient manner. We consider a threat model that reflects practical scenarios where all the parties are honest-but-curious or some of them collude to attack a particular victim. The main contribution of the paper is summarized as follows:
\begin{itemize}
    \item We develop a communication protocol that allows a training worker (i.e., the training process runs on a particular private dataset) to inform the master (i.e., the training coordinator that has the global model instance) about the evolution of the model parameters without revealing any sensitive information. The protocol also minimizes the communication overhead, i.e., the amount of data exchanged among the master and workers during the training. 
    
    \item We develop a goodness function that quantifies the impact of a private dataset on the global model instance, thus determining how much the parameters of the global model instance should be updated based on the information obtained from the respective training worker. Since the information from different training workers may not provide the same update direction for a certain model parameter, the goodness function also allows us to determine the update direction for a certain parameter.  We assume that the number of data samples in each private dataset is arbitrary: one may have a large number of samples while the other has much fewer samples. We also assume that the data samples in all the private datasets are stationary and have the same distribution. This relieves us from the concept of data drift and domain adaptation problems. 
    
    \item We develop a parallel and synchronous training algorithm at the master that invokes the training algorithm on each worker and wait for all the workers completing a training epoch before updating the global model instance. 
    
    \item We provide a formal analysis of the convergence of the training process as well as the privacy protection of the proposed framework. We show that the combination of the communication protocol and the goodness function adds non-linear factors to the information exchanged among the master and workers, thus making it hard to infer useful information about the private dataset of each worker.
    
    \item We carry out extensive experiments with two deep learning problems: image classification and image segmentation. We use two existing deep learning models for the experiments: ResNet50 FIXUP~\cite{Fixup_ZhangDM19} trained on CIFAR-$10$ dataset and U-Net~\cite{olaf:2015} trained on the LGG Segmentation Dataset~\cite{BUDA2019218}. We compare the performance (i.e., accuracy) and efficiency (i.e., communication overhead) of the proposed framework with two existing works~\cite{fedavg2017} and~\cite{phong2019}. 
\end{itemize}

The rest of the paper is organized as follows. In Section~\ref{sec:related_work}, we discuss the related works. In Section~\ref{sec:framework}, we present the design of the framework and the algorithms executed at the master and workers. In Section~\ref{sec:theory_analysis}, we provide a formal analysis of convergence and privacy-preserving properties of the framework. In Section~\ref{sec:experiment}, we present the experiments and analysis of results before we conclude the paper in Section~\ref{sec:conclusion}.

\section{Related Work}
\label{sec:related_work}

\subsection{Distributed Deep Learning} 

Conventional training techniques such as gradient descent, stochastic gradient descent, deep descent derivatives are based on an iterative algorithm that updates the model parameters as the iterations proceed to converge to specific values say \textit{optimal parameters}. The convergence speed (the number of iterations) depends on the learning rate and convergence condition set by the algorithm and usually requires a long-running time to complete. Many works in the literature have tried to speed up the convergence by considering different aspects including distributed storage of training data (i.e., data parallelism), distributed operation of computational tasks (i.e., model parallelism). In~\cite{LI201695}, Li \textit{et al.} analyzed different parallel frameworks for training large-scale machine learning and deep learning models. Example frameworks include Theano\footnote{Theano:~\url{http://deeplearning.stanford.edu/wiki/index.php/Neural_Networks}}, Torch~\cite{TU19961225}, cuda-convnet and cuda-convnet2\footnote{cuda-convnet2:~\url{https://code.google.com/p/cuda-convnet2/}}, Decaf~\cite{pmlr-v32-donahue14}, Overfeat~\cite{sermanet2013overfeat}, and Caffe~\cite{Jia:2014}. Most of these frameworks are open-source and optimized by NVIDIA GPUs using the Compute Unified Device Architecture (CUDA) programming interface. Similarly, in~\cite{Jia2019}, Jia \textit{et al.} introduced a comprehensive Sample-Operator-Attribute-Parameter (SOAP) search space of parallelization strategies for deep neural networks, namely FlexFlow, executing on multiple GPUs or CPUs. While those frameworks achieve good acceleration compared to computation on CPU architectures, their speedup is compromised by the exponential increase in the amount of data processed by the system.  

Focusing on reducing the number of computation steps, in~\cite{lecun2015}, Zhang \textit{et al.} proposed an algorithm, namely Elastic Averaging Stochastic Gradient Descent (EASGD), that trains a neural network in a distributed manner with both synchronous and asynchronous modes. The authors demonstrated that the proposed algorithm is better than DOWNPOUR~\cite{Dean:2012}. However, they assumed that all training workers have the entire dataset, thus ignoring the privacy issue. Further, both EASGD and DOWNPOUR do not consider the communication cost since the entire gradients and models are exchanged among the training workers. Also using a distributed computing approach, in~\cite{Le:2011}, Le \textit{et al.} proposed to use the Map-Reduce framework to enable parallel operations during the training of a deep neural network. 

\subsection{Communication Efficiency}
Focusing on reducing communication cost, Lin \textit{et al.}~\cite{Lin:2018} proposed a deep gradient compression approach to reduce the communication bandwidth between multiple training nodes. The gradients are compressed after a sequence of a pre-defined number of mini-batches. The authors defined a threshold to select which gradient values to be sent. He \textit{et al.}~\cite{He2018a} and Wen \textit{et al.}~\cite{Wen2017} used a ternary concept to represent the update direction of model parameters or gradients, i.e., $-1$ for decreasing parameters, $0$ for unchanged parameters and $+1$ for increasing parameters. We adopt this ternary approach in our work to minimize the communication overhead between training workers and the master but with an enhancement to protect data privacy. We further develop a goodness function that quantifies the impact of each private dataset on the global model, thus updating the model accordingly.

Alistarh \textit{et al.}~\cite{Alistarh2017-QSGD} defined the algorithm of gradients compression, named Quantized SGD (in short: QSGD). This approach firstly flatters the gradients into a 1-dimensional vector and then divides it into $k$ buckets with size $n$. The bucket is later compressed by a quantizer function to reduce data size to $\sqrt{n}(log(n + \mathcal{O}(1))$ before being transmitted over the network. Recently, Reisizadeh \textit{et al.}~\cite{reisizadeh20a-FedPAQ} applied this work for the federated learning, named FedPAQ, to reduce data transmission. However, before applying this approach, this work locally updates the model $T$ iterations (named communication delay) to reduce the communication overhead between the server and computing nodes. 

Recently, Sattler \textit{et al.}~\cite{Sattler2020} proposed a method, named Sparse Ternary Compression (STC), that compresses both downstream and upstream communication as well as working with non-iid data. This approach firstly extends the top-k sparsification to get the sparse ternary tensor in the flattened tensor $T^* \in \{-\mu, 0, \mu\}$ where $\mu$ is the average value of top-k values. Afterward, they used the Golomb code~\cite{Golomb-TIT.1966} to compress the position of non-zero elements (i.e., $\mu$'s value) in $T^*$ before transferring on the network. To decode this tensor, the authors also need a one-bit tensor to represent the sign of $\mu$'s value in $T^*$. 
 
\subsection{Privacy-Preserving}

In~\cite{phong2019}, Phong and Phuong proposed a framework that allows multiple data owners to train a deep learning model over the combined dataset to obtain the best possible learning output without sharing the local dataset, owing to privacy concerns. The authors developed two network topologies for the exchange of model parameters: a master-worker topology and a fully-connected (peer-to-peer) topology. The communication among training peers, master and workers is secured by the standard Transport Layer Security (TLS). However, this work does not enable parallel training as the deep learning model is sequentially trained on the private datasets. 

In~\cite{Shokri2015}, Shokri \textit{et al.} designed a system in which each local training worker asynchronously shares a part of local gradients obtained on its dataset to the master/server. However, while the training data is not shared among the training workers, an honest-but-curious attacker (who can be the master or a training worker) can collect the exchanged gradient values and infer the nature of the training dataset of a particular training worker. Further, this work does not consider secure communication to protect gradients against a man-in-the-middle attacker as well as communication cost. In~\cite{Phong2018}, Phong \textit{et al.} revisited the work of Shokri \textit{et al.}~\cite{Shokri2015} in combination with additively homomorphic encryption to protect the gradients against the honest-but-curious master. Hao \textit{et al.}~\cite{Hao2019} has a similar work that considers a larger number of training workers and integrates additively homomorphic encryption with differential privacy to protect data privacy. In~\cite{Tang2019}, Tang \textit{et al.} developed a two-phase re-encryption technique to protect the privacy of gradients exchanged among training workers. The authors adopted a Key Transfer Server (KTS) and a Data Service Provider (DSP) that issue a Diffie Hellman (private) key for each of the training workers to encrypt the gradients before sending them to the master. This allows the training workers to securely communicate with the master without establishing an independent secure channel. 

There also exist several works that used differential privacy to protect the training data and gradients. This approach defends against not only honest-but-curious actors (e.g., training workers and master) but also adversarial attackers that carry out model-inversion attacks. In~\cite{Fredrikson:2015}, Fredrikson \textit{et al.} developed a deferentially private stochastic gradient descent algorithm to protect the weight parameters from the strong adversary with full knowledge of the training mechanism. Later on, a similar technique has been developed in~\cite{Abadi2016}. 
In~\cite{Papernot2016} and~\cite{Papernot2018}, Papernot \textit{et al.} proposed an approach, namely Private Aggregation of Teacher Ensembles (PATE) to provide strong privacy guarantees for training data. This approach first combines multiple models (called \textit{teacher} models) trained with disjoint datasets in a black-box fashion. The teacher models are then used to predict the label of incomplete public data. Finally, the predicted public data is used to train \textit{student} models which will be released to use. In~\cite{Lyu2019}, Lyu \textit{et al.} used blockchain technology to decentralize the learning process where each training worker does not trust any third party or other training workers. In~\cite{Jones2018}, Beaulieu-Jones \textit{et al.} developed an approach for distributed deep learning for clinical data. The authors adopted differential privacy by adding noise to gradients so that no individual patient's data has a significant influence on the global model, thus revealing the nature of the respective training data. Similar to~\cite{phong2019}, this work sequentially trains the model from one private dataset to the other, i.e., the weights obtained on one dataset are transferred to the next training worker to continue on a new dataset. 

Secure Multi-Party Computation (MPC) is an approach in which multiple parties reveal information only about the final result, and not any of the input data. In order to hide the parties' model weight from the server, the encrypted model weight of computing parties is sent to the server. To encrypt the model weight, in the works of Mugunthan \textit{et al.}~\cite{Mugunthan2019SMPAISM} and Byrd \textit{et al.}~\cite{10.1145/3383455.3422562}, every pair of computing parties ($P_{i},P_{j}$) share a secret random value $r_{i,j}$. When communicating with the server, one party will add this value to its model weight and the other party will subtract it from its model weight. To avoid the case of the conclusion of $N-1$ parties where the weight of the remaining party can be easily determined, these works also combine differential privacy with MPC to make the system to be fully private. However, if the server updates the weight by the ratio of the data size of the party, these approaches may not be useful since the party does not know the ratio itself. Compared to differential privacy, MPC increases the communication cost of parties since they have to communicate either by pair or together to share the generated secure information that has been used to encrypt the model before sending it to the server. Recently, Sotthiwat \textit{et al.}~\cite{9499372} proposed a method that reduces this communication cost amongst parties since only the gradients of the first layer are encrypted with MPC strategy, while the rest are sent directly to the server.

\subsection{Federated Learning}
Recently, federated learning~\cite{Yang:2019} has emerged, requiring a novel distributed learning framework that allows the model training to be performed over geographically-distributed (private) datasets. Federated learning allows the data owners to collaboratively train a model without sharing its data with others, thus protecting data privacy. \textit{FederatedAveraging} (FedAvg)~\cite{fedavg2017} is one of the typical algorithms that was designed to train a machine learning model using multiple distributed datasets in a parallel manner. At each training step, a random subset of clients (i.e., training worker or data owner) is chosen and asked to train the model using their own dataset in parallel. Then, the clients send back the model parameters to the server to update the global model using an average formula. However, it does not consider reducing the data communication cost among clients and the server as well as data privacy. Recently, the work of Li \textit{et al.}~\cite{Li_FEDAVG_DP_2020} addresses data privacy by adding the differential privacy technique. At each training step, a client chooses either to update the model parameters or adding noise to the model parameters. Our work differs from Li \textit{et al.}~\cite{Li_FEDAVG_DP_2020}, as the computing nodes (clients) do not always send their model to the server, which may reveal data information. We further add the complexity and non-linearity to model parameters by using different training parameters for the computing nodes, such as learning rate, mini-batch size, epochs, training step, etc.

\section{Federated Deep Learning Framework for Privacy Preservation and Communication Efficiency (FedPC)}
\label{sec:framework}

In this section, we develop our proposed framework, namely \textbf{FedPC} - a \textbf{Fed}erated Deep Learning Framework for \textbf{P}rivacy Preservation and \textbf{C}ommunication Efficiency. First, we present an overview of the framework. Then, we describe the detailed algorithms executed at the master and workers during the training process.

\subsection{Framework Overview}

\begin{figure}
    \centering
    \includegraphics[width=0.48\textwidth]{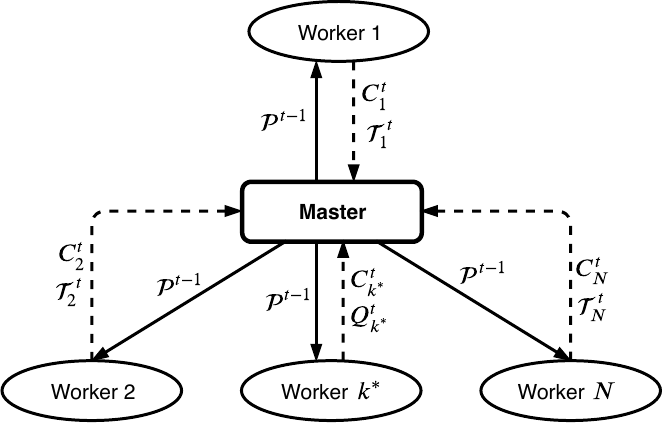}
    \caption{Architecture of FedPC.}
    \label{fig:arch}
\end{figure}

In Figure~\ref{fig:arch}, we present the architecture of the proposed framework, which consists of a master and  $N$ workers. The master is responsible for coordinating the training process: invoking the execution of the training algorithm on the workers, receiving training outcomes (presented in details in the next sections) from the workers, updating the global model instance and notifying the workers after the global model instance has been updated and ready for the next training epoch. We assume that the master maintains a convergence condition such that the global model will converge to the desired instance after a finite number of epochs. We denote this parameter as \texttt{GLOBAL\_EPOCHS}. 

The workers are the data owners equipped with a computing server to run a training algorithm on their private data. Given a model instance downloaded from the master, each worker runs the training algorithm and send training outcomes to the master once completed, using the communication protocol defined by the master. We assume that each worker is responsible for determining the hyper-parameters used in the training algorithm such as learning rate, mini-batch size, number of training epochs and optimizer. These parameters are the private information of each worker and they are kept unknown to the master as well as other workers, thus preventing others from inferring useful information about the training data. We also note that the size of the datasets of the workers is heterogeneous. Thus, with the heterogeneity in hyper-parameters and the size of datasets, the training time will be different for different workers. In this work, we implement a synchronous-parallel training approach such that the master will wait for all the workers to finish their local training before updating the global model instance and invoking a new epoch towards \texttt{GLOBAL\_EPOCHS}. 

\subsection{Master: Analyzing Training Outcomes from Workers and Updating Global Model}
\label{sec:master_update}

We now present the details of the algorithm executed at the master. Without loss of generality, we assume that the algorithm has been run for $t-1$ epochs and we are now at the beginning of the $t$-th epoch. In Algorithm~\ref{alg:master}, we present the steps that the master will execute in an epoch. The input of Algorithm~\ref{alg:master} is the model instance obtained from the previous iteration, denoted as $\mathcal{P}^{t-1}$. We assume that the model $\mathcal{P}$ has $M$ parameters, each being denoted as $P_i$ such that $\mathcal{P}^{t-1}$ is defined as  $\mathcal{P}^{t-1} =\{P_i^{t-1} | i =1\dots M\}$.

At the first step (line~\ref{line:invoke_workertraining}), the master invokes the execution of the training algorithm on each training worker. As discussed in the previous section, each worker has its own training algorithm with private hyper-parameters and optimization methods. Each worker downloads the model instance ($\mathcal{P}^{t-1}$) to its local memory and starts the training algorithm. Due to the heterogeneous size of datasets and hyper-parameters of the training algorithm on different workers, the training on the workers may not complete at the same time. The master needs to synchronize the execution by waiting for all the workers to complete their training before proceeding to the next step. 
\setlength{\textfloatsep}{16pt}
\begin{algorithm}[t]
\caption{\texttt{masterTraining}($\mathcal{P}^{t-1}$)}
\label{alg:master}
\begin{algorithmic}[1]
\State Invoke Alg.~\ref{alg:worker}: \texttt{workerTraining}($\mathcal{P}^{t-1}$) on all the workers\label{line:invoke_workertraining}
\State \texttt{synchronize}() \Comment{wait for all the workers to complete}
\State Receive cost from workers: $\{C_k^t|k=1\dots{}N\}$
\State Compute the goodness of the local instances using Eq.~\eqref{eq:goodness}
\State Receive the local model instance $Q_{k^*}^t$ from worker $k^*$
\State Receive ternary vector $\mathcal{T}_k^t$ of worker $k, k=1\dots{}N, k\neq k^*$
\State Update the global model using Eq.~\eqref{eq:master_update}
\State \Return $\mathcal{P}^t$
\end{algorithmic}
\end{algorithm}

Given that all the workers have finished their training, each worker obtains a local model instance denoted as $\mathcal{Q}_k^t$ ($\mathcal{Q}_k^t = \{Q_{k,i}^t | i = 1\ldots{}M\}$) and evaluates with their training dataset resulting in a cost, denoted as $C_k^t$ where $k=1\dots{}N$. The cost of the model could be the loss function value, reconstruction error, such that the lower the value of the cost, the better the model instance obtained. Instead of sending the local model instance to the master, which may not be as good as the one from other workers and also incurs high communication overhead due to its large size, each worker first sends its respective cost to the master. We define a goodness function to determine the best local model instance and its respective worker based on its cost.
\begin{equation}
\label{eq:goodness}
    G_k^t = \begin{cases}
    S_k\dfrac{1}{C_k^t}& \text{if}\; t=1,  \\
     S_k(C_k^{t-1} - C_k^t)& \text{if} \;t>1
    \end{cases}
\end{equation}
where $S_k$ is the size of the dataset of worker $k$. For instance, $S_k$ can be defined as the number of images in the dataset or the number of rows in a tabular dataset. The rationale behind the goodness function is that we consider the correlation between the cost and size of the dataset of each worker as a low cost could be resulted by a small dataset. At the first epoch, the worker with the lowest cost per data sample ($C_k^t/S_k$) is selected. Thus the goodness value of each model instance obtained from a worker after the first epoch is defined as the inverse of the per-unit cost.  From the second epoch onward, the progress in the training will also be taken into account. Since the master has the information of the cost from the previous epoch, it can evaluate the training progress by computing the reduction in cost. A worker with a large number of data samples and a large reduction in cost will be obviously selected. However, if a worker with less data but still has a large reduction in training cost, it can be also a good candidate to update the global model instance. The worker with the highest value of the goodness function, denoted as $G_{k^*}^t$ is selected as the pilot worker to send its local model instance denoted as $\mathcal{Q}_{k^*}^t$ to the master for updating the global model instance.  

We note that the master will not only use the local model instance $\mathcal{Q}_{k*}^t$ of worker $k^*$ to update the global model instance but also the training outcomes from other workers. Due to the privacy concern, those unselected workers will not send the entire model instance to the master but only the evolution direction of each parameter of the model in the form of a ternary vector. For worker $k$, the ternary vector obtained at epoch $t$ denoted as $\mathcal{T}_k^t$ is defined as follows:
\begin{equation}
    \mathcal{T}_k^t = \{T_{k,i}^t\}, \; i = 1\ldots{}M,\; k= 1\ldots{}N
\end{equation}
where $T_{k,i}^t \in \{-1, 0, 1\}$ indicates the change of parameter $P_i$ of the model produced by worker $k$ at epoch $t$ in correlation with that of epoch $t-1$. Let $Q_{k,i}^t$ denote the value of parameter $P_i$ obtained by worker $k$ at epoch $t$. The value of ternary vector means as follows: 
\begin{itemize}
    \item $T_{k,i}^t = -1$: If the change direction of parameter $P_i$ obtained by worker $k$ at epoch $t$ is different with that at epoch $t-1$. In other words, $T_{k,i}^t = -1$ if parameter $P_i$ increased at epoch $t-1$ and decreases at epoch $t$ (i.e., $Q_{k,i}^{t-2}<Q_{k,i}^{t-1}$ and $Q_{k,i}^{t-1}>Q_{k,i}^{t}$) or vice versa (i.e., decreased at epoch $t-1$ and increases at epoch $t$).   
    \item $T_{k,i}^t = 0$: If parameter $P_i$ does not significantly change for two consecutive epochs.
    \item $T_{k,i}^t = 1$: If parameter $P_i$ obtained by worker $k$ keeps significantly increasing or decreasing the same as at epoch $t-1$. In this case, $Q_{k,i}^t > Q_{k,i}^{t-1} > Q_{k,i}^{t-2}$ or $Q_{k,i}^t < Q_{k,i}^{t-1} < Q_{k,i}^{t-2}$ should hold. The former condition means that the parameter keeps increasing in two consecutive epochs while the later means that the parameter keeps decreasing in two consecutive epochs.
\end{itemize}

We note that computation of this ternary vector is performed by the workers, each having a different learning rate to determine whether or not parameter $P_i$ has significantly changed and at with direction it has changed. We will describe more detail in the next section. 

Given the local model instance received from worker $k^*$ and ternary vectors from the workers other than worker $k^*$, the master updates the global model instance as follows:
\begin{equation}
\label{eq:master_update}
    P_{i}^{t} = \begin{cases}
    Q^t_{k^*,i} - \alpha_0 \displaystyle\sum_{k \neq k^*} p_{k}T_{k,i}^t{} &\text{if}\; t = 1, \\
    Q^t_{k^*,i} - \displaystyle\sum_{k \neq k^*} p_k\beta_k{}T_{k,i}^t(P_i^{t-1} - P_i^{t-2}) &\text{if}\; t > 1
    \end{cases}
\end{equation}
where $\alpha_0$ is the learning rate of the master, $S_{k}$ is the data size of worker $k$, $S$ is the total data size of all workers, $S=\sum_{k=1}^NS_k$, $p_k =S_k/S$ (the proportion of data size of worker $k$), and $\beta_k$ is a parameter that the master synchronizes with worker $k$ to determine whether a model parameter has changed significantly on the dataset owned by worker $k$. Without loss of generality, the master can set the same value for all $\beta_k$'s. After updating all model parameters at epoch $t$, the master notifies all the workers to download the newly-obtained model instance ($\mathcal{P}^t$) and invokes a new training epoch ($t+1$). 

\subsection{Workers: Training Model and Ternarizing Evolution of Model Parameters}

In this section, we present the algorithm run at the workers, which train the model on their respective dataset, compute the ternary vector and update the master. As all the workers execute the same algorithm, without loss of generality, we present the algorithm for worker $k$. Algorithm~\ref{alg:worker} describes the pseudo-code of the algorithm. The input of the algorithm is the model instance received from the master after epoch $t-1$ and the outputs of the algorithm are a local model instance, cost value, and the ternary vector. 

\begin{algorithm}[t]
\caption{\texttt{workerTraining}($\mathcal{P}^{t-1}$)}
\label{alg:worker}
\begin{algorithmic}[1]
\State $\{C_k^t, \mathcal{Q}_k^t\}\leftarrow$ \texttt{trainModel}($\mathcal{P}^{t-1}$)
\State Send cost $C_k^t$ to the master
\State \texttt{CMD}$\leftarrow$\texttt{getMasterCommand}()
\If {\texttt{CMD}==\texttt{SEND\_MODEL}}
\State \Return $\mathcal{Q}_k^t$ \Comment{Send model instance $\mathcal{Q}_k^t$ to the master}\label{line:sendmodel}
\Else
\If{\texttt{CMD}==\texttt{SEND\_TERNARY}}
\State Compute ternary vector using Eq.~\eqref{eq:ternary1} or Eq.~\eqref{eq:ternary2}
\State \Return $\mathcal{T}_k^t$ \Comment{Send ternary vector to the master}
 \EndIf   
\EndIf 
\end{algorithmic}
\end{algorithm}

The worker procedure starts by running the training algorithm to obtain a local model instance and the training cost applied to the dataset owned by the worker. After the training completes, the worker sends the cost to the master for goodness evaluation and waits for the command to determine the next action. With the cost from all the workers, the master evaluates the goodness of the model instance obtained by each worker as described in the previous section. If the model instance obtained by worker $k$ is the best, it will send the model instance to the master and complete the algorithm (line~\ref{line:sendmodel} in Algorithm~\ref{alg:worker}). 

If the master requests the model instance from a different worker, worker $k$ has to compute the ternary vector and sends it to the master. At the first epoch ($t=1$), the ternary value $T_{k,i}^t$ of parameter $P_i$ obtained by worker $k$ is defined as follows: 

\begin{equation}
\label{eq:ternary1}
    T_{k,i}^{t} = 
    \begin{cases}
    -1 &\text{if}\; Q^{t}_{k,i} - P_i^0 < -\alpha_{k}, \\
    0 &\text{if}\; |Q^{t}_{k,i} - P_i^{0}| \leqslant \alpha_{k}, \\
    1 &\text{if}\; Q^{t}_{k,i} - P_i^0 >\alpha_{k} \\
    \end{cases}
\end{equation}
where $\mathcal{P}^{0}=\{P_i^{0} | i =1\dots M\}$ is a model instance randomly initialized by the master for the first epoch, and $\alpha_k, k =1\ldots{}N$ is the learning rate of worker $k$. The rationale behind this equation is that at the first epoch, the workers do not have the evolution history of each model parameter. The workers can only compute the ternary vector based on their learning experience through the learning rate. If a parameter significantly decreases compared to its initialized value, the ternary value will be $-1$. If the change in the parameter value is not significant $|Q^{t}_{k,i} - P_i^{0}| \leqslant \alpha_{k}$, we can say that the parameter did not change. If the parameter significantly increases, the ternary value will be set to $1$.

From the second epoch onward ($t\geqslant2$), we define the ternary vector based on the evolution history of the model. We assume that the workers keep a copy of the model instance received from the master in epoch $t-1$ and epoch $t-2$, denoted as $\mathcal{P}^{t-1}$ and $\mathcal{P}^{t-2}$, respectively. The ternary value $T_{k,i}^t$ of parameter $P_i$ obtained at worker $k$ and at epoch $t$ is defined as follows:
\begin{equation}
\label{eq:ternary2}
T_{k,i}^{t} = 
    \begin{cases}
    0 &\text{if}\; |Q^{t}_{k,i} - P_i^{t-1}| < \beta_{k}|P_i^{t-1} - P_i^{t-2}|, \\
    \texttt{sign}(f)  &\text{otherwise} \\
    \end{cases}
\end{equation}
where $f = (Q^{t}_{k,i} - P_i^{t-1})(P_i^{t-1} - P_i^{t-2})$ is the production of the changes at two previous epochs. 

As explained in the previous section, $T_{k,i}^t = -1$ if parameter $P_i$ has changed in different directions at two previous epochs. $T_{k,i}^t = 0$ if parameter $P_i$ has not significantly changed at epoch $t$ compared to that at epoch $t-1$. $T_{k,i}^t = 1$ if parameter $P_i$ has significantly changed in the same direction for two consecutive epochs. It is worth recalling that $\beta_k$ is a parameter that the master provides to worker $k$ to determine whether  or not a model parameter has significantly changed. $\beta_k$'s can be set to a number in the range $(0,1)$, (e.g., $0.2$).  

By using a ternary vector with values that belong to $\{-1,0,1\}$ to represent the evolution of the model parameters from the workers, we can reduce the communication cost from the workers to the master by up to $32\times$, compared to the case of sending the entire model when using a $64$-bit number for each parameter. This is because we can represent these three values by $2$ bits (e.g., $00$, $01$ and $11$). Thus, we can compress $4$ ternary values into $1$ Byte. Even using a $32$-bit number for each, we also reduce the communication cost by up to $16\times$.
 
\section{Convergence and Privacy Analysis}
\label{sec:theory_analysis}

\subsection{Convergence Analysis}
In this section, we present a formal proof that our proposed model updating approach as presented in Eq.~\eqref{eq:master_update}, Section~\ref{sec:master_update} guarantees the convergence of learning models.

\begin{thm}[Convergence of learning models] 
\label{thm:convergence}
By applying the updating approach presented in Eq.~\eqref{eq:master_update}, Section~\ref{sec:master_update} for model parameters, a learning model trained by FedPC converges to the optimal instance after a sufficiently large number of epochs. 
\end{thm}

\begin{proof}
We refer the reader to \ref{sec:proofthm1} for a formal proof of the theorem.
\end{proof}

\subsection{Privacy Analysis}

Privacy is one of the most important properties of federated learning. In this section, we prove that our framework guarantees data privacy and does not reveal any useful information even for intermediate computational results to different parties participating in the training process and external attackers. The privacy analysis in this paper is partially based on the privacy aspects used in ~\cite{phong2019}.

\subsubsection{Threat model}
In this work, we consider two attacking scenarios coming from both insider attacks and external attacks.
\begin{itemize}
    \item Insider attacks: We assume that all the master and workers participating in the training process are honest-but-curious. If any information from the training process is leaked, they will exploit and infer useful information for their purpose. We also assume that the workers do not trust each other and want to protect their training data against other data owners. In an extreme case, we consider an attack scenario where maximum $N-2$ workers are malicious and collude together to attack a particular worker where $N$ is the total number of workers participating in the training process.
    \item External attacks: We consider a man-in-the-middle attack scenario where an attacker sniffs the information exchanged among the master and workers during the training process and infer the useful information. We assume that on each system of the master and workers, there will be existing security solutions that prevent illegal intrusion to steal model or data.
\end{itemize}

We note that while we consider an honest-but-curious threat model in this work, it can happen that in reality, a compromised worker can choose to not follow the protocol and send adversarial information to the master. This leads to the fact that the learning model is no longer trustful and it can result in wrong predictions. However, this does not affect data privacy, i.e., the privacy of the training data, which is the main focus of our work. On one hand, the attacker at the compromised worker still cannot learn anything about the training data from other workers. On the other hand, the attacker cannot obtain a better learning model compared to the model trained on its local data, taking the advantage of the large training data from other workers, thus defeating the purpose of federated learning. It is also worth mentioning that addressing this challenging threat deserves a separate work. A possible solution is to adopt the Byzantine problem that has been extensively studied in networking~\cite{9039724}. With the Byzantine approach, the master (which is a trustful party) can verify the information received from workers before updating the model. Another possible approach is to randomly drop out several workers at each training epoch so as to avoid bias to a particular worker.

\subsubsection{Analysis}

\begin{thm}[Privacy against the honest-but-curious master] 
\label{thm:against_master}
By applying mini-batch gradient descent with private training parameters on the workers and the proposed approach for updating the model parameters at the master, an honest-but-curious master has no information on the private dataset of the workers, unless it solves a non-linear equation.
\end{thm}

\begin{proof}
Assuming that the workers use mini-batch gradient descent to optimize the local model on its dataset, Phong and Phuong have successfully proven part of this theorem in~\cite{phong2019}. For completeness, we briefly present below.   

Given $n$ as the number of batches of the dataset of worker $k$ and $n \gg 1$, a model parameter is updated after each epoch as follows:
\begin{equation}
    Q_{k,i}^t = Q_{k,i}^{t-1} - \alpha_k (G_{1} + ... +G_{n}).
\end{equation}
Since batch size (and thus number of batches, $n$) and learning rate $\alpha_k$ are private information of worker $k$ and unknown to the master, individual gradient on each mini-batch is therefore also unknown to the master. Nevertheless, the master can compute the sum of gradients if the master has the model parameters for two consecutive epochs.
\begin{equation}
    Q_{k,i}^{t-1} - Q_{k,i}^{t} = \alpha_k(G_{1} + ... +G_{n}).
    \label{eq:grad-sum}
\end{equation}  

Recovering any data item of the dataset of worker $k$ from equation \eqref{eq:grad-sum} is the task of solving a non-linear equation. In other words, This is a subset sum problem.
This proves the theorem.
\end{proof}

\begin{thm}[Privacy against honest-but-curious workers] 
\label{thm:against_collusion}
With the proposed framework, FedPC, an honest-but-curious worker has no information on the private dataset of other workers.
\end{thm}

\begin{proof}
Let us refer to Eq.~\eqref{eq:master_update} that is used by the master to update the model parameters after each epoch. Worker $k$ can infer the dataset of work $k^*$ if and only if the two following scenarios happen at the same time:
\begin{enumerate}
    \item Worker $k$ knows that the master has requested worker $k^*$ to send its local model instance for a sufficiently large number of epochs; and
    \item While worker $k^*$ sends its local model instance to the master, all other workers send the all-zeros-ternary vectors to the master. This makes the global model instance at the master be exactly the local model instance obtained at worker $k^*$. 
\end{enumerate}
Assuming that the above conditions hold, worker $k$ can infer the private dataset of worker $k^*$ if worker $k$ has the capability of solving a non-linear equation system as discussed previously. However, this assumption could not happen in practice when the master and workers are honest. Given the equal importance of the datasets owned by different workers, the ternary vector provided by the workers evolves during the training process. This proves the theorem. 
\end{proof} 

\begin{thm}[Privacy against a collusion of up to $N-2$ workers] 
The proposed framework, FedPC, preserves the privacy of the private dataset of a particular worker against collusion of up to $N-2$ workers.
\end{thm}
\begin{proof}
As discussed in the proof of Theorem~\ref{thm:against_collusion}, which is broken when $N-1$ workers are malicious and they collude together to retrieve information of the private dataset of the remaining worker (victim worker). Note that $N$ is the number of workers participating in the training process. This can be achieved when these colluding workers always send the master the  unchanged training cost in comparison to the previous train (i.e., the goodness equals zero.), making the master always request the remaining worker to send its local model instance. Further, these colluding workers will send all zeros for the ternary vector when requested by the master. Hence, the updated parameter based on the equation \eqref{eq:master_update} is basically contributed by the victim worker alone. Given a sufficiently large number of training epochs, these colluding workers can form a non-linear equation system to solve and infer the private dataset of the remaining worker. It is to be noted that in this extreme case, the honest-but-curious master can also get benefits as it maintains all the model instances requested by the victim worker. 

We now assume that there are a maximum of $N-2$ workers that are malicious and collude together. The two remaining workers are benign. After each training epoch, the master requests a local model instance from one of these two honest-but-curious workers. Given the equal importance of datasets owned by the two benign workers, the malicious workers will not know which worker is requested to send its local model instance. Further, the remaining benign worker could also send a non-zero ternary vector as its local model should evolve (converge) along with the training process. In this case, the honest-but-curious master is also not able to infer the private dataset of a particular worker as the two benign workers could be alternatively requested for their local model instance after each epoch. This proves the theorem. 
\end{proof}

\paragraph{Discussion}

The privacy scenario in the Theorem \ref{thm:against_master} is similar to the one of Theorem 2 in Section 4 using SNT system in~\cite{phong2019}. Although the master could get the model parameters for two consecutive epochs from worker $k$, the later model parameter could be involved in the dataset of other workers via their ternary vectors. Hence, the non-linear equation turns out much more complicated to be solved.
      
Let us assume that the master has the model parameters from worker $k$ for $2(n+1)$ epochs with the condition that there should be $n+1$ pair of consecutive epochs, the master can form a non-linear equation system with $n+1$ equations as shown in Eq.~\eqref{eq:grad-sum}. Recovering individual gradient information from worker $k$ now becomes the problem of solving a non-linear equation system for $n+1$ variables (the learning rate of worker $k$ and gradients of $n$ training epochs) given that the master has such a capability. It is also worth mentioning that the number of batches ($n$) is unknown to the master, to achieve such a non-linear equation system, the master may assume a sufficiently large $n$ to cover all possibilities of batch size at the workers. However, the above assumption could not happen in practice due to the fact that the dataset at the workers are equally important. At each epoch, the master could request the model instance from a different worker based on the value of the goodness function presented in Eq.~\eqref{eq:goodness}. In other words, our approach for selection of model instance from the workers based on the goodness function prevents an honest-but-curious master from consecutively requesting the model instance from the same worker.

In order to protect data privacy from the honest-but-curious master and/or colluding workers in the scenario that the same worker is requested to send his/her trained model parameter for several consecutive epochs, the worker can make his/her own rules. For example, after \textit{a fixed number of steps}, if the global model instance received from the master at the beginning of each epoch is always identical to its local model instance obtained from previous epoch, the worker can
\begin{enumerate}
    \item apply existing techniques for protecting data privacy such as homomorphic encryption~\cite{Shokri2015,Phong2018,Tang2019}, and differential privacy~\cite{Fredrikson:2015} by adding noises to its local model instance before sending to the master, or
    \item stop sending the local model instance by simply sending the cost  unchanged from the previous request (namely, the goodness is equal to zero).  
\end{enumerate}
    

\section{Experiments}
\label{sec:experiment}
\subsection{Experiment Setup}

We implemented the proposed framework, FedPC, in Python using the TensorFlow framework. The communication between the master and workers for the invocation of training and synchronization was implemented using secure socket programming. The data transfer (model instances and ternary vectors) was implemented using secure copy protocol (SCP). We deployed the framework on several computers including:
\begin{itemize}
    \item Three customized desktop with AMD Ryzen Threadripper $2950$X $16$-core Processor @ $3.5$GHz, $64$ GB of RAM and $2$ GeForce RTX $2080$ Ti, each having $11$ GB of memory.  
    \item A server with $40$ cores of Intel(R) Xeon(R) Silver 4210R CPU @ 2.40GHz, $512$ GB of RAM and $8$ Quadro RTX $5000$, each having $16$ GB of memory.
\end{itemize}
We used $1$ customized desktop as the master and the remaining for deploying workers, each running on a GPU card (i.e., the server can run upto $8$ workers). 

To demonstrate the performance of the proposed framework, we carried out the experiments with two datasets including CIFAR-10~\cite{cifar10} and the LGG Segmentation Dataset~\cite{BUDA2019218}, denoted as LGGS, which is downloaded from its Kaggle Repository\footnote{LGGS Dataset:~\url{https://www.kaggle.com/mateuszbuda/lgg-mri-segmentation}}. We split LGGS Dataset to training and test set with a ratio $80:20$ for our experiments. We evaluated the proposed framework using different performance metrics:
\begin{itemize}
   
    \item Performance approximation: We trained the respective model in a centralized approach where all the training datasets are stored in the same location. The performance of the centrally-trained model (e.g., classification accuracy or segmentation accuracy) is considered as the upper bound to compute the performance approximation ratio of the model trained with the distributed approach using the proposed framework. 
    
    \item Amount of data exchanged during the training: We computed the amount of data exchanged among the master and workers per epoch during the training. This is an additional communication overhead which does not incur in the centralized training. The more the data exchanged, the longer the data transferring time and the higher the network bandwidth consumption.
\end{itemize}

We compare the performance of the proposed framework with two existing works: Phong \textit{et al.}'s method~\cite{fedavg2017} and FedAvg~\cite{phong2019} in two performance metrics: classification and segmentation accuracy of the models, and communication overhead.

\paragraph{Design of Deep Learning Models}
We used TensorFlow to implement deep learning models. The structures of existing deep neural networks have been adopted for respective datasets: ResNet50 FIXUP~\cite{Fixup_ZhangDM19} for CIFAR-$10$ dataset and U-Net~\cite{olaf:2015} for LGGS. We kept the design of neural networks (i.e., number of convolutional layers, number of filters, etc.) the same as the original design. We adopted the source codes of ResNet50 FIXUP from its GitHub repository\footnote{ResNet50 FIXUP:~\url{https://github.com/ben-davidson-6/Fixup/tree/master/tensorflow_implementation}}. For U-Net, the size of images is set to $256\times{}256$. We also applied padding for the convolution operations in U-Net.

\paragraph{Hyper-parameter Settings}

As we discussed earlier, hyper-parameters used for training the model at the workers are heterogeneous and private to them. The framework allows the workers to choose the batch size and local training epochs among a predefined list of values. For instance, the workers can take a batch size value among $128$, $64$ or $32$. For the LGGS dataset, the batch size can be selected among $16$, $8$ or $4$. We limit the upper bound due to the limited capacity of our GPU cards. The initial learning rate for all workers is set to $0.01$ but we apply the decay based on training steps, which is in turn based on the size of the local training dataset. This makes the learning rate of the workers also become heterogeneous after a few training epochs. We used Momentum optimizer~\cite{QIAN1999145} for ResNet50 FIXUP model and Adam optimizer~\cite{kingma:2015} for U-Net.

\subsection{Performance Analysis}
\subsubsection{Performance with Centralized Training}

We carried out the centralized training, which is similar to the training with only one worker that has the entire dataset. In Table~\ref{tab:centralized_training}, we present the accuracy of the models on different datasets. We achieved the expected state-of-the-art performance for each of the datasets. This shows that we have successfully implemented state-of-the-art models and reproduced the same performance. It is to be noted that there exist enhanced models that achieve higher performance for CIFAR-$10$. For instance, the work presented in~\cite{resnet32:2016} achieves a performance of $92.49\%$ of classification accuracy. In those models, multiple \texttt{BatchNorm} layers are used. However, the parameters of the \texttt{BatchNorm} layers computed on the entire dataset (in case of centralized training) is much different from those computed on individual private dataset of each worker. Sharing those parameters will reveal data information, especially mean and standard deviation of data distribution. Thus, we used a neural network architecture without \texttt{BatchNorm} layers. We use this performance as a (upper-bound) reference for comparison when evaluating the proposed framework (FedPC) that enables parallel training on distributed and private datasets. 
\begin{table}[t]
\caption{Performance with centralized training}  
\label{tab:centralized_training}
\centering
\begin{tabular}{lrr}
\hline
{\bf Dataset} & {\bf \#Epochs}&  {\bf Accuracy} \\

\hline
CIFAR-$10$ &$250$ & $0.9172 \pm 0.0008$\\

LGGS &$35$&  $0.9959 \pm 0.0004$ \\
\hline
\end{tabular}
\end{table} 
 
\subsubsection{Performance of FedPC}

\begin{figure}[t] 
    \centering 
    \includegraphics[width=0.48\textwidth]{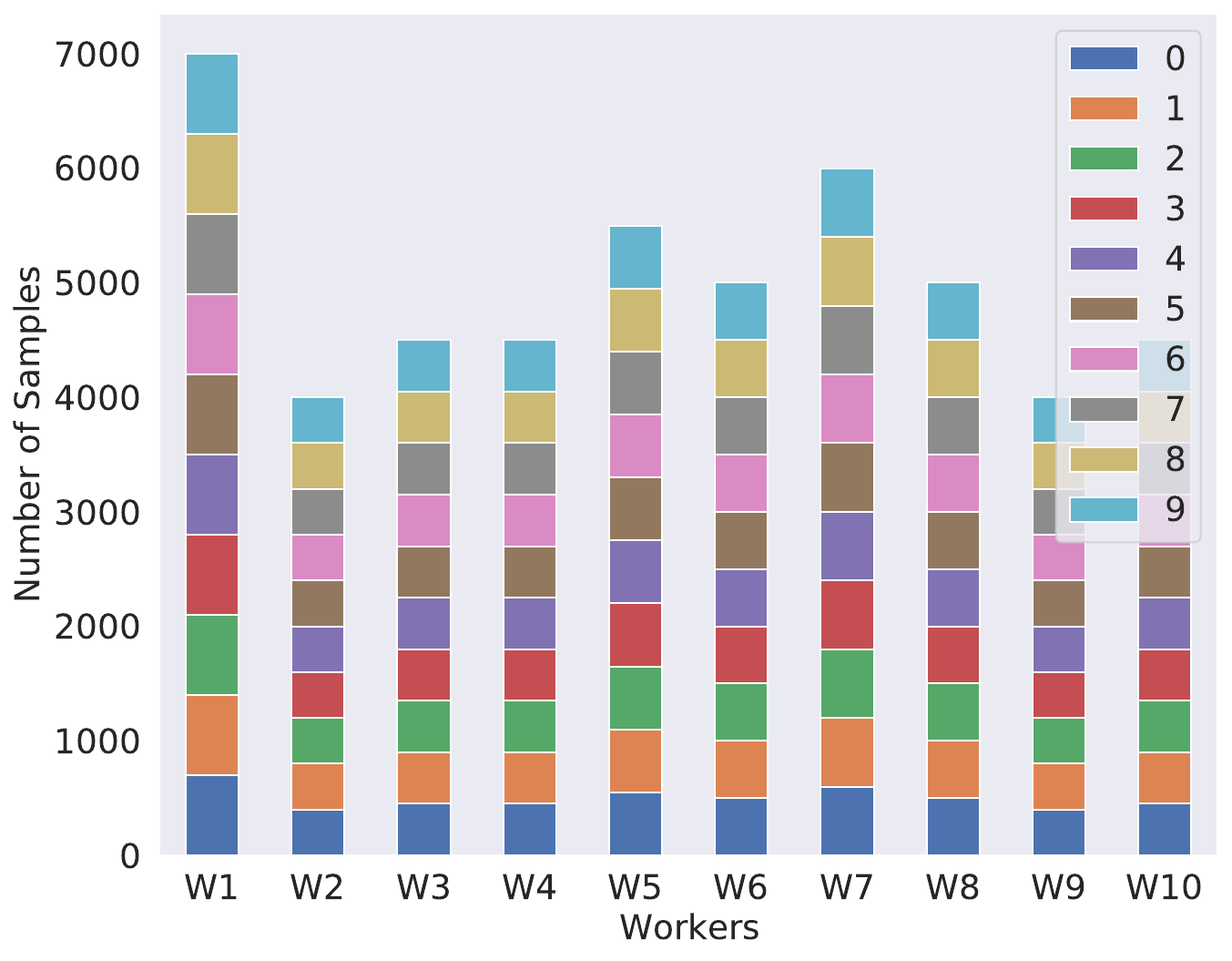}
    \caption{Data Distribution of CIFAR-$10$ among Workers.}
    \label{fig:data_dist_balance}
\end{figure}

\begin{figure}[t] 
    \centering 
    \includegraphics[width=0.48\textwidth]{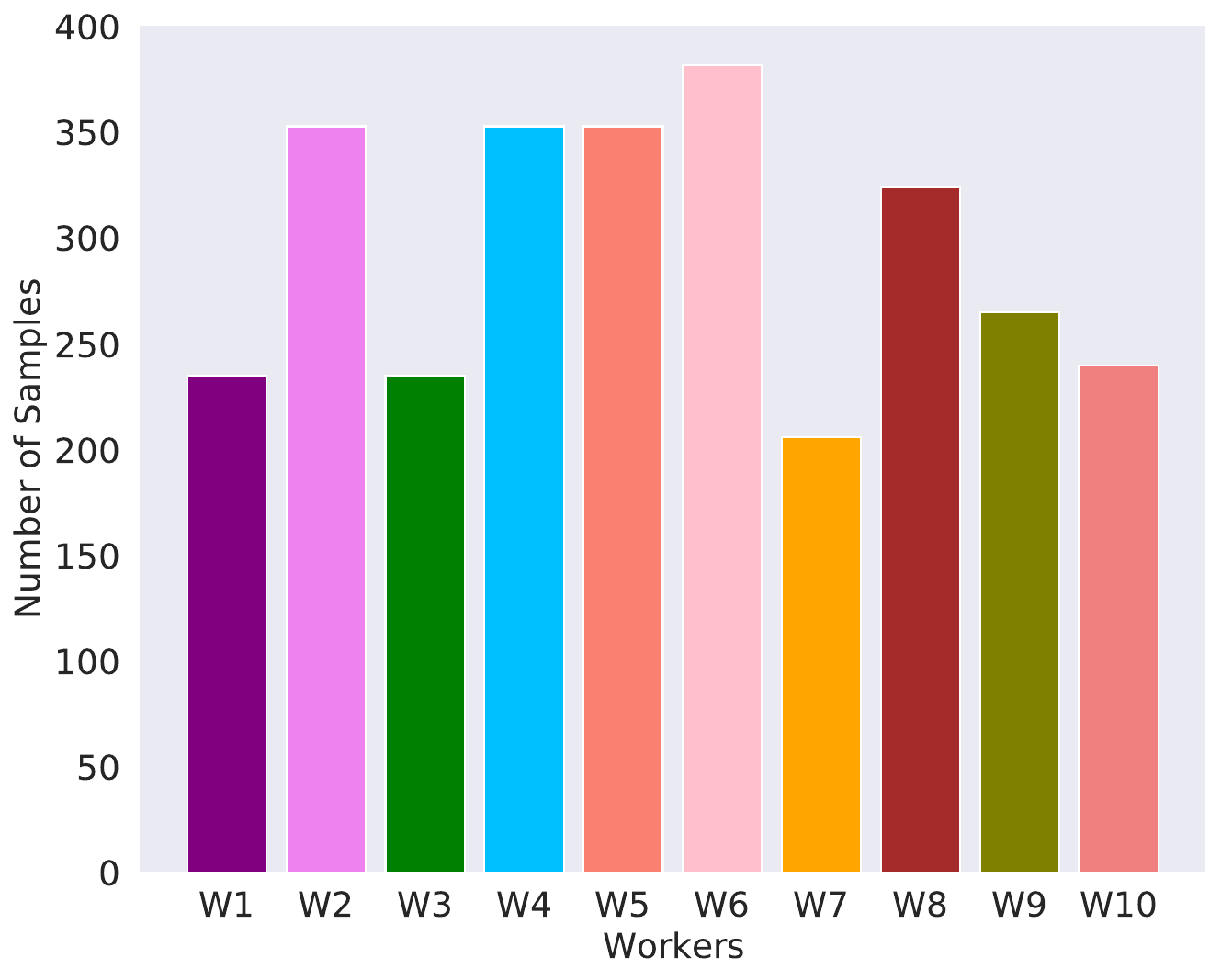}
    \caption{Data Distribution of LGGS among Workers.}
    \label{fig:data_dist_lggs}
\end{figure}  
 
To carry out parallel training on distributed datasets with FedPC, we randomly split the entire dataset into a specific number of portions and distributed them to a corresponding number of workers. We controlled the range of random variables to avoid the extreme imbalance scenarios of the split datasets, e.g., one worker has only $1\%$ of data while another worker has more than $90\%$ of data. We first randomly generate the percentage values, each corresponding to the data percentage that each worker will process. The percentage values sum to $100$. For CIFAR-$10$, the percentage values are then used to distribute the samples of each class of the dataset to various workers. This leads to the fact that the total number of data samples processed by each worker is heterogeneous/imbalanced.  But at each worker, the number of samples in each class is equal. An example of CIFAR-$10$ split into $10$ workers is shown in Figure~\ref{fig:data_dist_balance} where each color represents a sample class in the training dataset. For the LSSG dataset, as this is a segmentation problem, there is no data distribution at the class level. We show the distribution of the dataset among workers in Figure~\ref{fig:data_dist_lggs} where each color represents a worker. We carried out the experiments (parallel training) with an increasing number of workers. 

\paragraph{Performance Approximation}

In this section, we analyze the performance of FedPC in terms of the accuracy of the models trained by FedPC with an increasing number of workers. In Table~\ref{tab:accuapprox}, we present the accuracy of the models on the test set of CIFAR-10 dataset obtained with our proposed framework and the baselines.

\begin{table*}[t]
\caption{Accuracy Obtained with Different Algorithms on CIFAR-10 Test Set}
\label{tab:accuapprox}
\centering
\begin{tabular}{lrrr}
\hline
  {\bf Number of Workers} & {\bf FedPC} & {\bf Phong \textit{et al.}~\cite{phong2019}} & {\bf FedAvg~\cite{fedavg2017}} \\
\hline
$3$ & $0.9009 \pm 0.0031$ & $0.9101 \pm 0.0030$  & $0.9113 \pm 0.0014$ \\
\hline
$4$ & $0.8914 \pm 0.0030$ & $0.9119 \pm 0.0024$ & $0.9044 \pm 0.0053$ \\
\hline
$5$ & $0.8819 \pm 0.0024$ & $0.9115 \pm 0.0038$ & $0.9021 \pm 0.0030$\\
\hline
$6$ & $0.8752 \pm 0.0041$ & $0.9134 \pm 0.0034$ & $0.8992 \pm 0.0029$ \\
\hline
$7$ & $0.8697 \pm 0.0011$ & $0.9102 \pm 0.0028$ & $0.8935 \pm 0.0001$ \\
\hline
$8$ & $0.8632 \pm 0.0057$ & $0.9117 \pm 0.0026$ & $0.8906 \pm 0.0013$ \\
\hline
$9$ & $0.8471 \pm 0.0011$ & $0.9095 \pm 0.0010$ & $0.8824 \pm 0.0013$ \\
\hline
$10$ & $0.8388 \pm 0.0015$ & $0.9120 \pm 0.0011$ & $0.8772 \pm 0.0023$ \\
\hline
\end{tabular}
\end{table*} 

\begin{table*}[t]
\caption{Accuracy Obtained with FedPC on LGGS Test Set}
\label{tab:acculggs}
\centering
\begin{tabular}{lrrr}
\hline
  {\bf Number of Workers} & {\bf FedPC} & {\bf Phong \textit{et al.}~\cite{phong2019}} & {\bf FedAvg~\cite{fedavg2017}}\\
\hline
$3$ & $0.9952 \pm 0.0001$ & $0.9966 \pm 0.0001$ & $0.9959 \pm 0.0001$\\
\hline
$4$ & $0.9951 \pm 0.0003$ & $0.9913 \pm 0.0044$ & $0.9955 \pm 0.0005$\\
\hline
$5$ & $0.9887 \pm 0.0001$ & $0.9962 \pm 0.0001$ & $0.9926 \pm 0.0034$ \\
\hline
$6$ & $0.9915 \pm 0.0025$ & $0.9936 \pm 0.0042$ & $0.9944 \pm 0.0012$ \\
\hline
$7$ & $0.9914 \pm 0.0023$ & $0.9926 \pm 0.0055$ & $0.9942 \pm 0.0012$\\
\hline
$8$ & $0.9925 \pm 0.0010$ & $0.9912 \pm 0.0043$ & $0.9904 \pm 0.0019$\\
\hline
$9$ & $0.9901 \pm 0.0024$ & $0.9936 \pm 0.0043$ & $0.9938 \pm 0.0006$\\
\hline
$10$ & $0.9887\pm 0.0001$ & $0.9959 \pm 0.0005$ & $0.9887 \pm 0.0002$\\
\hline
\end{tabular}
\end{table*} 

As expected, we observed a performance drop when training a model in parallel on distributed datasets. Such a drop slightly increases along with the increase in the number of workers. With CIFAR-$10$ dataset, when training with $3$ workers, we observed a drop of $1.8\%$ compared to the performance of the centralized training approach. With the existing approaches (Phong \textit{et al.}~\cite{phong2019} and FedAvg~\cite{fedavg2017}), we also observed a slight performance drop of $0.6\%$ for both approaches. Such performance drop is due to the fact that the model training does no longer benefit from the data shuffling after each training epoch. Shuffling on a private (smaller) dataset does not create much variance to update the model parameters. While CIFAR-10 is a benchmarking dataset, real-world datasets may be more heterogeneous, e.g., due to the quality of data collection devices. Additionally, to protect the data privacy, the proposed approach does not exchange the model parameters (i.e., weights and bias of neurons) as does Phong \textit{et al.}~\cite{phong2019}'s method and FedAvg. This leads to further degradation of performance of the proposed approach compared to Phong \textit{et al.}~\cite{phong2019}'s method. It is interesting to note that performance degradation with Phong \textit{et al.}'s method keeps unchanged when increasing the number of workers, both the proposed framework (FedPC) and FedAvg incur larger performance drop even though FedAvg also exchanges model parameters. When the data is distributed to $10$ workers, the drop is $8.5\%$ for FedPC and $4.4\%$ for FedAvg.

In Table~\ref{tab:acculggs}, we present the performance of the proposed framework and the baselines on the LGGS dataset using a U-Net architecture. While we observed similar trends of performance, the performance degradation is very minor. In the worst scenario, the performance loss is only $0.7\%$. Compared to the performance loss incurred with the CIFAR-10 dataset, the performance loss incurred with the LGGS dataset is much smaller and could be negligible. This high performance could be explained by the fact that the quality of data stored at different computing nodes is not much different from each other. Furthermore, this is a segmentation problem, the performance of the segmentation model (U-Net) is not affected by the imbalance of the data among workers as well as at each individual worker. We believe that such a slight performance drop could be acceptable in regards to the gain of data privacy, thus demonstrating the effectiveness of our proposed framework.  

\paragraph{Performance Convergence}

\begin{figure}[t]
\centering
  \subfloat[Training Cost Evolution on CIFAR-10 Dataset.]
    {\label{fig:cifar-cost}\includegraphics[width=0.46\textwidth]{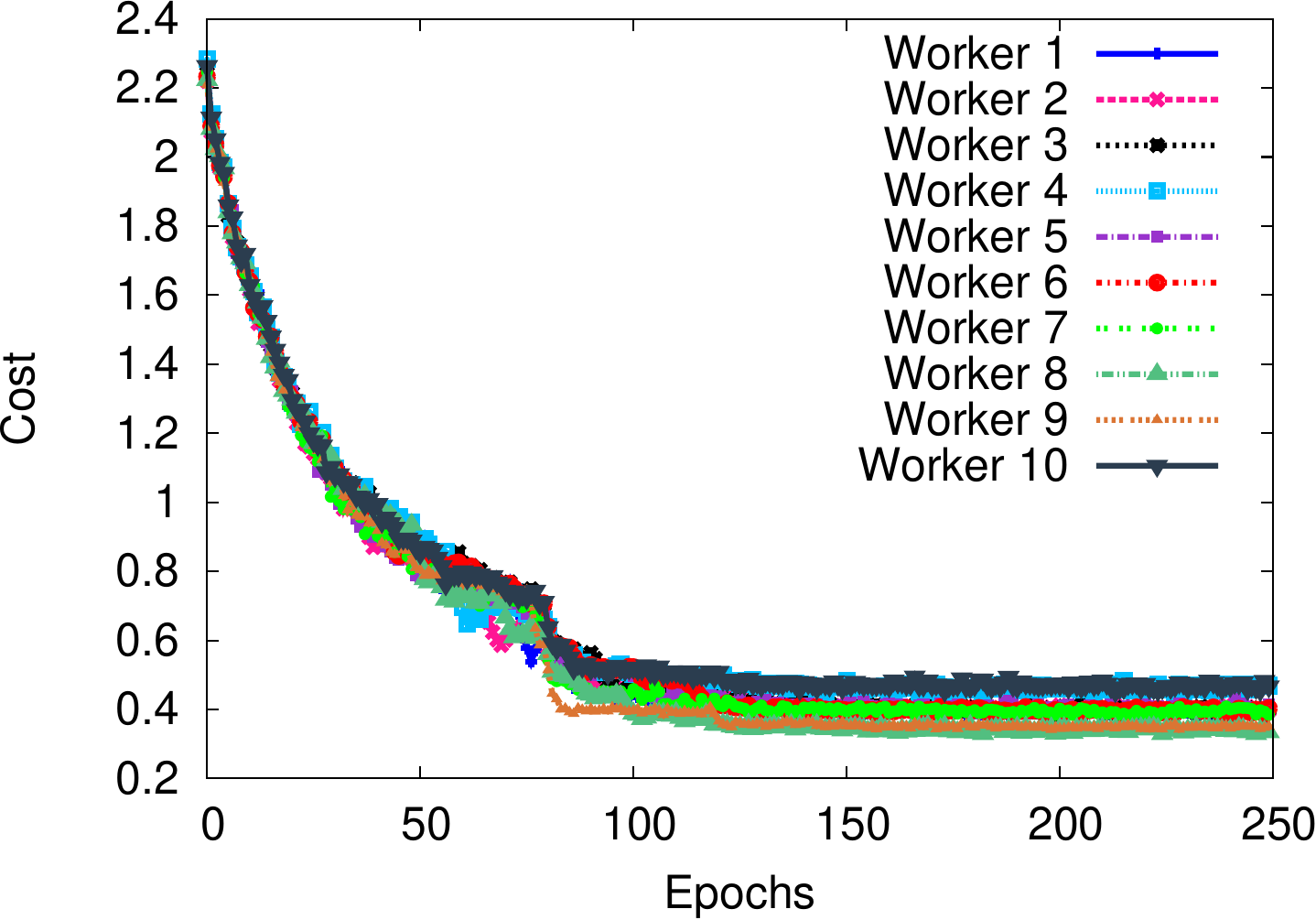}}\\
  \subfloat[Training Cost Evolution on LGGS Dataset.]
    {\label{fig:lggs-cost}\includegraphics[width=0.46\textwidth]{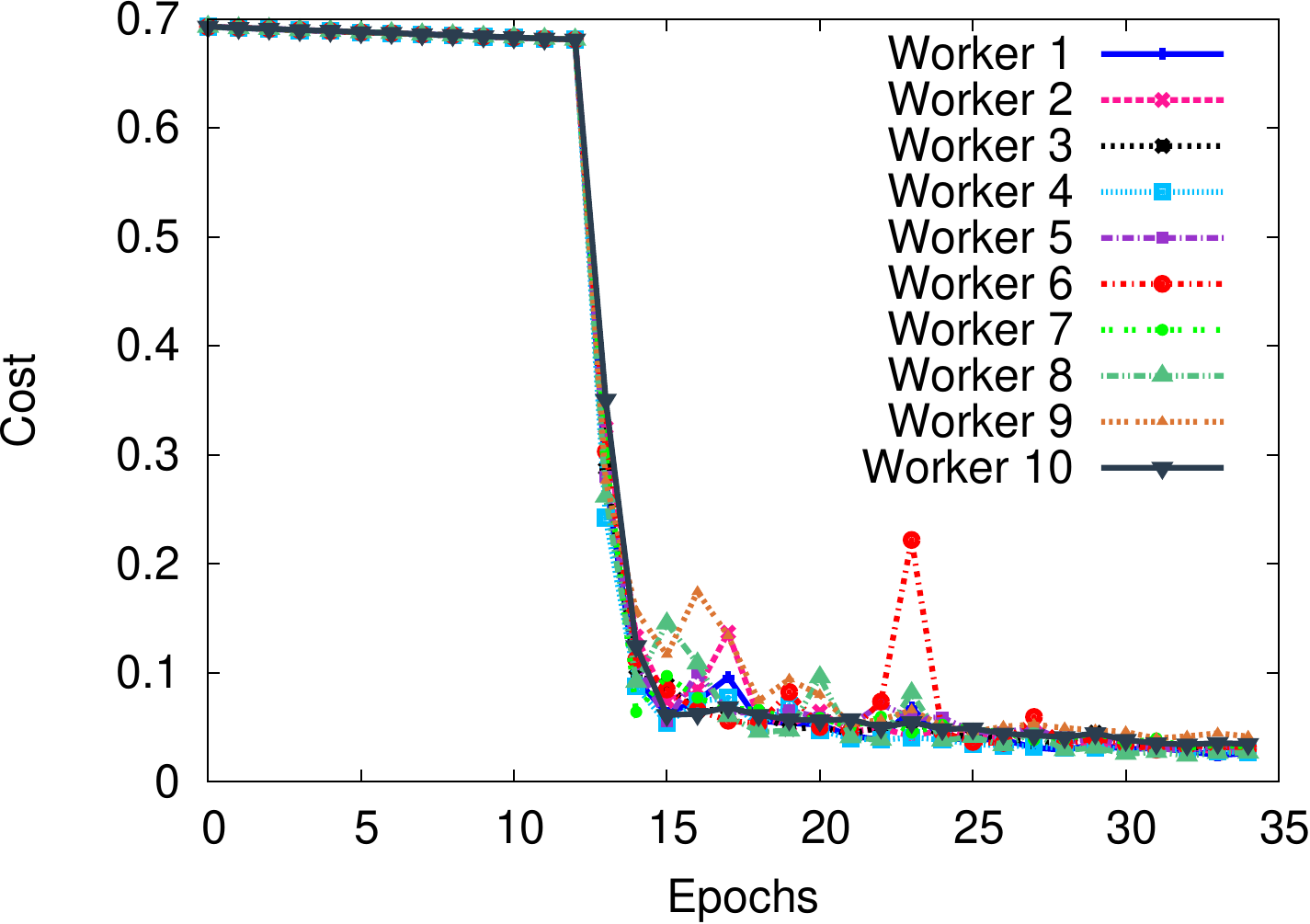}}
\caption{Training Convergence through Cost Reduction.}
\label{fig:convergence}
\end{figure}

While we provide mathematical proofs of model convergence in~\ref{sec:proofthm1}, we additionally analyze the convergence of the model training through the cost reduction along with training epochs. In Figure~\ref{fig:convergence}, we present the evolution of training cost of the models trained on the two datasets through training epochs. The experimental results show that the training cost gradually reduces and keep stable after a certain number of training epochs (e.g., $100$ epochs for CIFAR-10 and $25$ epochs for LGGS). Since the number of training epochs with CIFAR-10 is high, the plot does not clearly show the behavior for the first few epochs. However, we can observe such a behavior in Figure~\ref{fig:lggs-cost}. As such, the cost does not reduce significantly after the first few epochs before it starts converging to the minimum-achievable cost. This is explained by the fact that the workers in the proposed framework updates the master the evolution direction of model parameters rather than their absolute value or gradient. In other words, the master only know whether a parameter evolves in the same or oposite direction compared to that of the previous training epoch. Thus, the correct evolution can only be observed from the third training epoch onward. Combined with a small learning rate set for the training process, we obtained this behavior. Nevertheless, the training still converges, allowing the models trained with the proposed framework to approximate the performance of the models trained in the centralized manner. 



\paragraph{Performance with Non-identically Distributed Data}

\begin{figure}[t]
    \centering
    \includegraphics[width=0.48\textwidth]{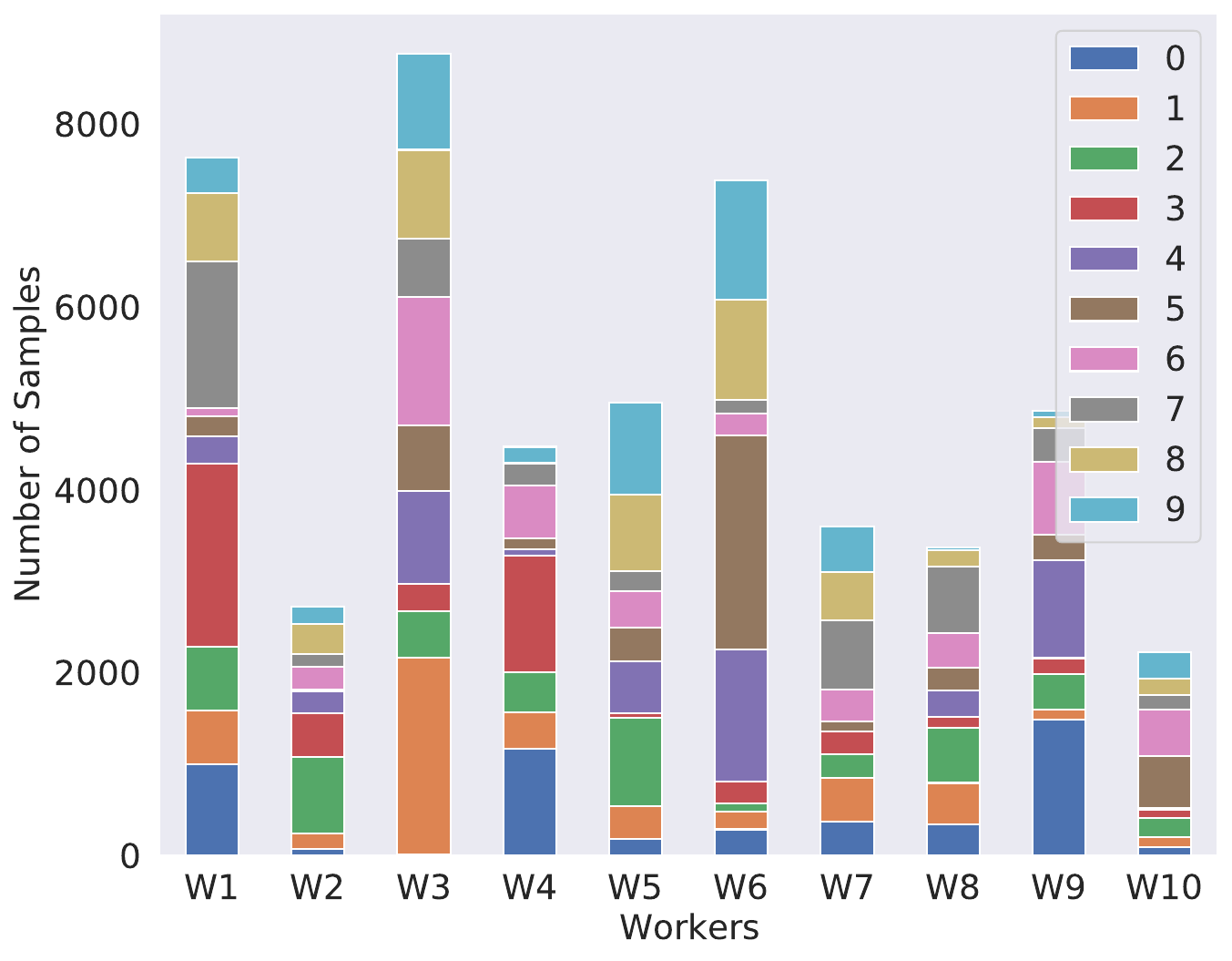}
    \caption{Data Distribution of CIFAR-$10$ among Workers using Dirichlet Distribution.}
    \label{fig:data_dist_noniid}
\end{figure}

In this experiment, we use the Dirichlet distribution~\cite{Minka00estimatinga} to split data among workers. Figure~\ref{fig:data_dist_noniid} shows an example of data distribution of CIFAR-$10$ with $10$ workers split by the Dirichlet distribution used for training ResNet50 FIXUP models. In Table~\ref{tab:accucifar10_noniid}, we present the performance of FedPC along with two baselines. Compared to the results shown in Table~\ref{tab:accuapprox}, the results of this experiment show that all the frameworks exhibit slight performance degradation when training with distributed data split by the Dirichlet distribution. This is reasonable for any multi-class classification models when training on an imbalanced dataset. Overall, we can observe a clear trade-off between data privacy and classification accuracy, FedPC provides the highest data privacy, thus achieving lower classification accuracy with the worst drop of $6.3\%$ and $10.4\%$ compared to FedAvg and Phong \textit{et al.}'s method, respectively. It is to be noted that this is affected only with multi-class classification problems when the data is very imbalanced at computing workers. This could be rarely happen in reality as each data owner could collect sufficient samples for each class. This demonstrates the broad practicality of the proposed framework to other deep learning problem such as image segmentation, object detection, etc.

\begin{table}[t]
\caption{Accuracy of FedPC Trained on Non-IID CIFAR-10}
\label{tab:accucifar10_noniid}
\centering
\begin{tabular}{lrrr}
\hline
  {\bf No. Workers} & {\bf FedPC} & {\bf Phong \textit{et al.}~\cite{phong2019}} & {\bf FedAvg~\cite{fedavg2017}}\\
\hline
$3$ & $0.8927$ & $0.9051$ & $0.9016$\\
\hline
$5$ & $0.8564$ & $0.8955$ & $0.8879$ \\
\hline
$7$ & $0.8317$ & $0.8941$ & $0.8751$\\
\hline
$10$ & $0.8001$ & $0.8934$ & $0.8540$\\
\hline
\end{tabular}
\end{table} 

\paragraph{Amount of Data Exchanged Among Master and Workers}

We have measured the amount of data exchanged among the master and workers during a training epoch. We compared our proposed approach with two existing works. The total amount of data exchanged between the master and workers depends on the size of the models, size of ternary vectors and the number of workers. For every training epoch, our framework needs to copy an initial model from the master to all the workers. After completing the epoch, one worker needs to send its local model instance to the master and other $N-1$ workers need to send the ternary vector, whose size is $16\times$ smaller than that of the model instance. The total amount of data exchanged between the master and workers per epoch can be computed as follows:
\begin{equation}
    D = V(N+1) + \dfrac{V(N-1)}{16}
\end{equation}
where $V$ is the size of a model instance (i.e., $35$ MB for a ResNet$50$ FIXUP model instance and $119$ MB for a U-Net model instance) and $N$ is the number of workers. We note that we used \texttt{float32} data type for all the model parameters when implementing the training framework.

\begin{figure}[t]
    \centering
    \includegraphics[width=0.46\textwidth]{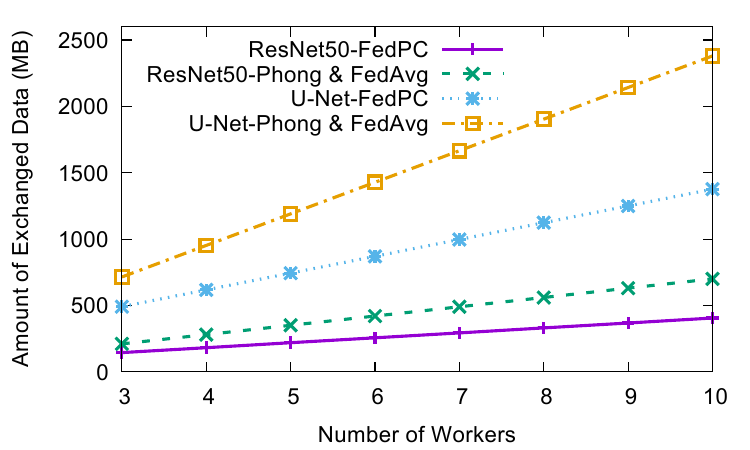}
    \caption{Amount of Data Exchanged in Each Training Epoch.}
    \label{fig:exchangeddata}
\end{figure}

Phong \textit{et al.}'s method~\cite{phong2019} and FedAvg~\cite{fedavg2017} require every worker to download the initial model instance and upload its local model instance to the master after each epoch, leading to a total amount of data of $2\times{}V\times{}N$. In Figure~\ref{fig:exchangeddata}, we present the amount of data exchanged among the master and workers per epoch when training the models in parallel with different numbers of workers. The results show that our proposed framework (FedPC) significantly reduces the amount of data exchanged among the master and workers. Compared to Phong \textit{et al.}'s method and FedAvg, FedPC reduces the amount of data by at least $31.25\%$ for both models. The improvement is much more significant when there are more workers deployed in the framework, i.e., data is distributed to many more geographical locations. Indeed, when there are $10$ workers, the proposed approach reduces the amount of data by upto $42.20\%$. The total amount of data exchanged to complete the model training linearly increases along with the number of training epochs. For a fair comparison, we kept the same number of training epochs for both centralized training and distributed training of all algorithms, i.e., 250 epochs for ResNet50 FIXUP and 35 epochs for U-Net models. The experimental results show that our framework (FedPC) incurs the least communication overhead and offers the highest data privacy while achieving good performance approximation to the centralized training.

\paragraph{Study with a Real-world Application}

\begin{figure}[t]
\centering
\begin{minipage}[b]{.5\textwidth}
  \subfloat[Original Image.]
    {\label{fig:habitg1_histo}\includegraphics[width=0.46\textwidth]{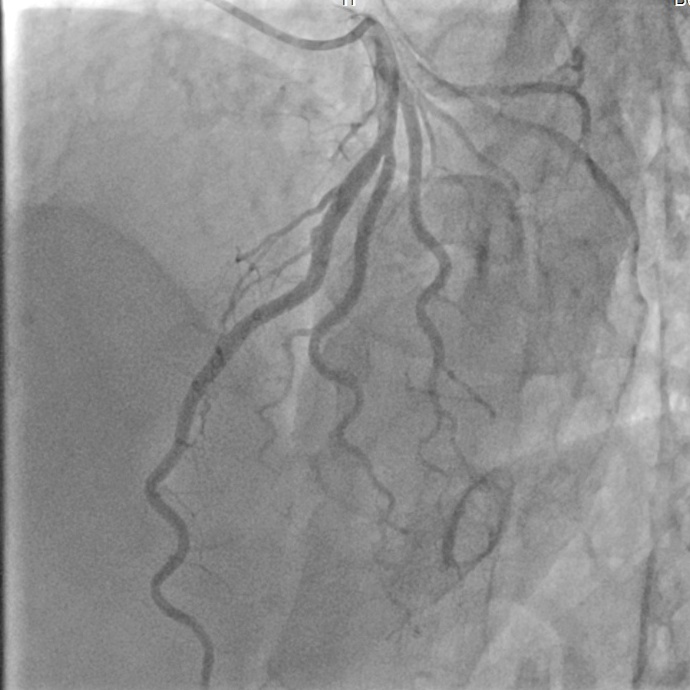}}\hspace{2mm}
  \subfloat[Ground Truth Label.]
    {\label{fig:habitg2_histo}\includegraphics[width=0.46\textwidth]{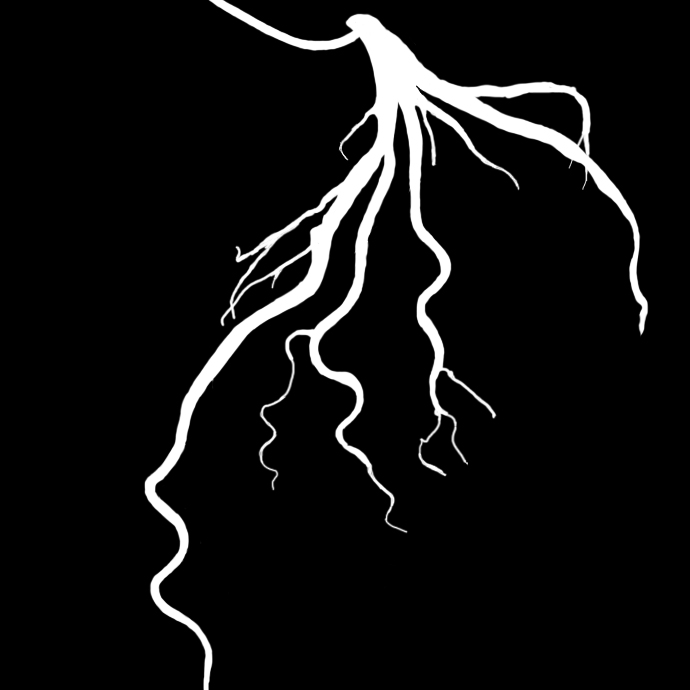}}
    \\
  \subfloat[By Centralized Model.]
    {\label{fig:habitg3_histo}\includegraphics[width=0.46\textwidth]{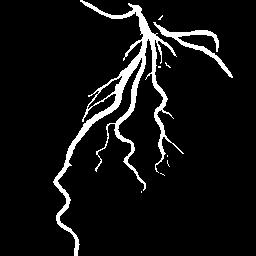}} 
    \hspace{2mm}
  \subfloat[By Model trained with $3$ Workers.]
    {\label{fig:habitg4_histo}\includegraphics[width=0.46\textwidth]{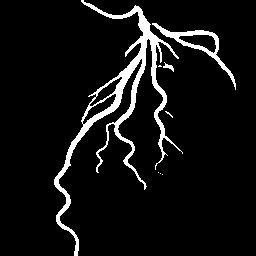}}
\end{minipage}
\caption{Segmentation Results with U-Net on HEART-VESSEL.}
\label{fig:illustration}
\vspace{-1.5ex}
\end{figure}

We further studied the performance of the proposed framework, FedPC, with a real-world application used at the School of Medicine, Tan Tao University. The application is also an image segmentation problem applied to heart vessels. The dataset includes $45$ heart vessel images, stored in $3$ different locations. We deployed our framework with $3$ workers at the three storage locations. The experimental results show that the centralized training approach achieves a segmentation accuracy of $0.9214$ while the model trained with the distributed approach with $3$ workers approximates the performance with an accuracy of $0.9210$. For illustration, in Figure~\ref{fig:illustration}, we present the segmentation results obtained by different models and refer to the segmentation obtained by the clinician. We can see that the segmentation obtained by the models trained with FedPC is almost the same as the ground truth label (i.e., segmentation performed by the clinician).




\section{Conclusion}
\label{sec:conclusion}
In this paper, we designed and developed FedPC, a federated deep learning framework for privacy preservation and communication efficiency. FedPC allows a deep learning model to be trained in parallel on geographically-distributed datasets, which are assumed to be in the same data distribution. We defined a goodness function that helps the master evaluate the importance of the model instances obtained at the workers. The master holds a global model instance and updates it after each training epoch. We developed a ternarizing approach that allows the workers to inform the master about the evolution of the model on their local dataset without revealing the gradients and data samples. We formally proved the convergence of the models trained by our framework and the privacy-preserving properties guaranteed. We extensively carried out experiments with popular deep learning models that are used for diverse problems such as image classification and segmentation. The experimental results show that our proposed framework achieves the data protection without sacrificing the performance of the models (i.e., accuracy). FedPC maintained a performance degradation less than $8.5\%$ of the model trained with the centralized approach. Compared to existing approaches, our framework reduced the amount of data exchanged among the master and workers by up to $42.20\%$ when training the models with $10$ workers. 

\section*{Acknowledgements}
This work was partially supported by Tan Tao University Foundation for Science and Technology Development under the Grant No. TTU.RS.19.102.005.

\appendix

\section{Proof of Theorem~\ref{thm:convergence}}
\label{sec:proofthm1}
\begin{proof}
In \cite{Zinkevich:2010}, the authors have proved that in a regular distributed-learning framework, a learning model will converge to the optimal instance if all the workers send gradient obtained on their local dataset to the master. With the gradient, a model will be updated as follows:
\begin{equation}
    \mathcal{P}^t = \mathcal{P}^{t-1} - \dfrac{1}{S}\displaystyle\sum_{k=1}^N\beta_kS_k\mathcal{G}_{k}^{t-1}
    \label{eq:gradientupdate}
\end{equation}
where $S_k$ is the size of dataset at worker $k$, $S$ is the sum of all private datasets at the workers, $S=\sum_{k=1}^N{}S_k$, and $\mathcal{G}_{k}^{t-1}$ is the average gradient at a data point at epoch $t-1$ from worker $k$, i.e.,
\begin{equation}
    \mathcal{G}_{k}^{t-1}= \dfrac{1}{b}\displaystyle\sum_{j\in D_{k,b}}\nabla{L}(\mathcal{P}^{t-1}, X_j)
\end{equation}
where $D_{k,b}$ is a data batch and $X_j$ is a data point. Given a certain worker $k^*$ and let $p_k = S_k/S$, we can rewrite Eq.~\eqref{eq:gradientupdate} as follows: 
\begin{align}
    \mathcal{P}^t  &=  \mathcal{P}^{t-1} - p_{k^*}\beta_{k^*}\mathcal{G}_{k^*}^{t-1} - \displaystyle\sum_{k\neq{k^*}}p_k\beta_k{}\mathcal{G}_{k}^{t-1}\nonumber\\
    & = \mathcal{Q}_{k^*}^t - \displaystyle\sum_{k\neq{k^*}}p_k\beta_k{}\mathcal{G}_{k}^{t-1}
\end{align}
where 
\begin{align}
    \mathcal{Q}_{k}^t & =  \mathcal{P}^{t-1} - p_{k}\beta_{k}\mathcal{G}_{k}^{t-1} 
     = \mathcal{P}^{t-1} - \beta_{k}\mathcal{G}_{k}^{t-1} + (1-p_{k})\beta_{k}\mathcal{G}_{k}^{t-1} \nonumber\\ 
    & = \mathcal{P}_{k}^t + (1-p_{k})\beta_{k}\mathcal{G}_{k}^{t-1}.
\end{align}
$\mathcal{P}_{k}^t$ can be considered as the updated model from worker $k$ based on its dataset with learning rate $\beta_k$ while $\mathcal{Q}_{k}^t$ is the one at the master and be updated with the learning rate $p_k\beta_k$.

We approximately have $|\mathcal{C}(\mathcal{P}_{k}^t) - \mathcal{C}(\mathcal{P}^{t-1})| \cong L||\mathcal{P}_k^t - \mathcal{P}^{t-1}||$ where $\mathcal{P}_k^t$ is the model instance obtained by worker $k$ at epoch $t$ and $\mathcal{P}^{t-1}$ is the global model instance at the master, $\mathcal{C}$ is the cost function. Thus, the goodness of the model instance obtained by worker $k$ can be approximately computed as follows:
\begin{equation}
p_k|\mathcal{C}(\mathcal{P}_k^t) - \mathcal{C}(\mathcal{P}^{t-1})| \cong L{}p_k||\mathcal{P}_k^t - \mathcal{P}^{t-1}|| = L{}p_k\beta_k||\mathcal{G}_k^{t-1}||.
\end{equation}

Determining worker $k^*$ that has the highest goodness value of the model instance is thus solving the following problem:
\begin{equation}
    k^* = \arg\max_k(p_k\beta_k||\mathcal{G}_k^{t-1}||) = \arg\max_k(||\mathcal{Q}_k^t - \mathcal{P}^{t-1}||).
\end{equation}
Note that $||\mathcal{Q}_k^t - \mathcal{P}^{t-1}||=p_k||\mathcal{P}_k^t-\mathcal{P}^{t-1}||$.

Let us now consider a voted framework in which the global model instance at the master will be updated based on the model instance obtained by worker $k^*$. We have, 
\begin{equation}
    \mathcal{P}^t = \mathcal{Q}_{k^*}^t - \displaystyle\sum_{k\neq{k^*}}p_k\beta_k{}\mathcal{G}_{k}^{t-1} \cong \mathcal{Q}_{k^*}^t - \displaystyle\sum_{k\neq{k^*}}p_k\beta_k{}\mathcal{T}_{k}^t(\mathcal{P}^{t-1} - \mathcal{P}^{t-2}).    
\end{equation}
It is to be noted that model $\mathcal{P}$ is defined as an $M$-dimensional vector with $M$ parameters $\{P_i, i=1\ldots{}M\}$. Each dimension is considered as a direction that the model will move to converge to the optimal instance. Parameter $P_i$ at time $t$ will be updated as follows:
\begin{align}
P_i^t & = Q_{k^*,i}^t - \displaystyle\sum_{k\neq{k^*}}p_k\beta_k{}G_{k,i}^{t-1} \nonumber\\
&\cong Q_{k^*,i}^t - \displaystyle\sum_{k\neq{k^*}}p_k\beta_k{}T_{k,i}^t(P_i^{t-1} - P_i^{t-2}).    
\end{align}
This means that for worker $k$ alone, $Q_{k,i}^t= P_i^{t-1} - p_k\beta_k{}G_{k,i}^{t-1}$ or $Q_{k,i}^t - P_i^{t-1} = - p_k\beta_k{}G_{k,i}^{t-1}$. If the change in any direction $i$ is not ``significant'' (which is less than a threshold), it would be ignored. Otherwise, it would be approximated by the threshold. We take the threshold to be some percentage of the previous step that a parameter has changed, which is $\beta_k|P_i^{t-1} - P_i^{t-2}|$. Thus, we have the ternary value $T_{k,i}^t$ defined as 
\begin{equation}
T_{k,i}^{t} = 
    \begin{cases}
    0 &\text{if}\; |Q^{t}_{k,i} - P_i^{t-1}| < \beta_{k}|P_i^{t-1} - P_i^{t-2}|, \\
    \texttt{sign}(f)  &\text{otherwise} \\
    \end{cases}
\end{equation}
where $f = (Q^{t}_{k,i} - P_i^{t-1})(P_i^{t-1} - P_i^{t-2})$ and $P_i^{t-1} - P_i^{t-2}$ is the step that parameter $P_i$ has moved in direction $i$ at time $t-1$. 
\begin{figure}[t]
\centering 
\begin{minipage}[b]{.5\textwidth}
  \subfloat[Change in direction $i$ by worker $k$ is significant (but $j$) and the change has the same sign as $(P_i^{t-1} - P_i^{t-2})$.]
    {\label{fig:proof_illustration_forward}\includegraphics[width=0.46\textwidth]{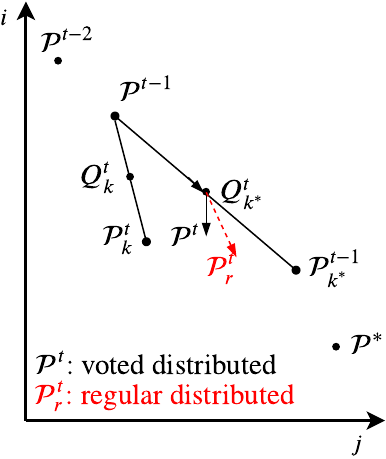}}\hspace{2mm}
  \subfloat[Change in direction $i$ by worker $k$ is significant (but $j$) and the change has the opposite sign as $(P_i^{t-1} - P_i^{t-2})$.]
    {\label{fig:proof_illustration_backward}\includegraphics[width=0.46\textwidth]{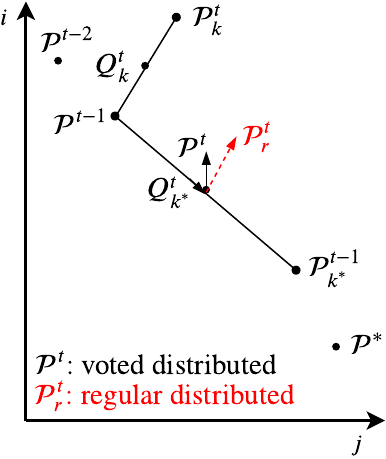}}
\end{minipage}
\caption{Illustrative Example for Training Convergence.}
\label{fig:proof_illustration}
\vspace{-0.5ex}
\end{figure}

In our approach, the selected worker $k^*$ is the one that reduces the cost most at epoch $t$ individually. It is well-known that the cost reduction is proportional to the step size that the model parameters move. Practically speaking, the distance $||\mathcal{Q}_{k^*}^t-\mathcal{P}^{t-1}||$ is larger than any other $||\mathcal{Q}_{k}^t-\mathcal{P}^{t-1}||$, thus model $\mathcal{P}^t$ is computed mainly based on $\mathcal{Q}_{k^*}^t$. In addition, we also take into account the change in direction $i$ from any worker $k$ other than worker $k^*$ if the change is significantly large. Thus, parameter $P_i^t$ of model $\mathcal{P}^t$ is updated along direction $i$ by $\beta_k$ of the previous step $|P_i^{t-1} -P_i^{t-2}|$ to ensure that it still points towards the optimal parameter $P_i^*$ of the optimal model $\mathcal{P}^*$. Specifically, if the change $(P_{k,i}^t-P_i^{t-1})$ is in the same direction as in the previous step $(P_{i}^{t-1}-P_i^{t-2})$, the point $\mathcal{Q}_{k^*}^t$ is pushed forward to $\mathcal{P}^*$ in direction $i$ as in Fig.~\ref{fig:proof_illustration_forward}. Otherwise, it is pushed backward by a small step $\beta_k|P_i^{t-1} - P_i^{t-2}|$ as illustrated in Fig.~\ref{fig:proof_illustration_backward}. This reflects the general trends when training a model that all workers have their own locally-updated model $\mathcal{P}_k^t$ pointing towards $\mathcal{P}^*$. The master simply somehow takes the weighted average of those $\mathcal{P}_k^t$'s as $\mathcal{P}^t$ then this $\mathcal{P}^t$ also points towards $\mathcal{P}^*$. We note that the illustrative example in Fig.~\ref{fig:proof_illustration} assumes that model $\mathcal{P}$ has two parameters. For a general case with $M$ parameters (each parameter is considered as a dimension), dimension reduction can be performed so that the convergence of the training can be visualized. This proves the theorem.
\end{proof}

\bibliographystyle{elsarticle-num}
\bibliography{jsa-cs4iot2021}	
\end{document}